\documentclass[11pt]{amsart}

\usepackage{amsmath}
\usepackage{amssymb}
\usepackage{amsfonts}
\usepackage{latexsym}
\usepackage{color}

\newtheorem{theo}{Theorem}[section]
\newtheorem{lem}[theo]{Lemma}
\newtheorem{prop}[theo]{Proposition}
\newtheorem{cor}[theo]{Corollary}
\newtheorem{lemma}[theo]{Lemma}

 \theoremstyle{definition}
\newtheorem{definition}[theo]{Definition}

\newtheorem{example}[theo]{Example}
 \theoremstyle{remark}

 \numberwithin{equation}{section}

\newtheorem{remark}[theo]{Remark}

\newcommand{\betheo}{\begin{theo}$\!\!\!${\bf } }
\newcommand{\entheo}{\end{theo}}

\newcommand{\becor}{\begin{cor}$\!\!\!$  }
\newcommand{\encor}{\end{cor}}

\newcommand{\belem}{\begin{lem}$\!\!\!${\bf .} }
\newcommand{\enlem}{\end{lem}}

\newcommand{\beprop}{\begin{prop}$\!\!\!${\bf } }
\newcommand{\enprop}{\end{prop}}

\newcommand{\bedefi}{\begin{definition}$\!\!\!$ \rm }
\newcommand{\findefi}{ \end{definition}}

\newcommand{\beex}{\begin{example}$\!\!\!$ \rm }
\newcommand{\enex}{ \end{example}}

\newcommand{\berem}{\begin{remark}$\!\!\!$ \rm }
\newcommand{\enrem}{ \end{remark}}

\newcommand{\be}{\begin{equation}}
\newcommand{\en}{\end{equation}}

\newcommand{\bea}{\begin{eqnarray}}
\newcommand{\ena}{\end{eqnarray}}

\newcommand{\beano}{\begin{eqnarray*}}
\newcommand{\enano}{\end{eqnarray*}}

\newcommand{\bee}{\begin{enumerate}}
\newcommand{\ene}{\end{enumerate}}

\newcommand{\bei}{\begin{itemize}}
\newcommand{\eni}{\end{itemize}}

\newcommand{\betab}{\begin{tabular}}
\newcommand{\entab}{\end{tabular}}

\newcommand{\bd}{\begin{displaymath}}

\newcommand{\h}{{\mathfrak H}}

\newcommand{\ad}{^{\mbox{\scriptsize $\dag$}}}
\newcommand{\das}{^{\dag {\rm\textstyle *}}}

\newcommand{\up}{\raisebox{0.7mm}{$\upharpoonright $}}%

\newcommand{\pow}[1]{^{\mbox{\scriptsize $(#1)$}}}

\def\B{{\mathcal B}}

\def\D{{\mathcal D}}
\def\E{{\mathcal E}}
\def\F{{\mathcal F}}

\def\H{{\mathcal H}}

\def\L{{\mathcal L}}
\def\M{{\mathcal M}}
\def\N{{\mathcal N}}

\def\J{\relax\ifmmode {\mathcal J}\else${\mathcal J}$\fi}
\def\x{\relax\ifmmode {\mbox{*}}\else*\fi}

\def\MM{{\mathfrak M}}

\newcommand{\mc}{\mathcal}
\newcommand{\mb}{\mathbb}

\newcommand{\mult}{{\scriptstyle \Box}}

\newcommand{\pa}{partial \mbox{*-algebra}}

\newcommand{\LD}{{\L}\ad(\D)}
\newcommand{\LDH}{{\L}\ad(\D,\H)}

\def\dag{\dagger}

\newcommand{\ip}[2]{\left\langle {#1}\left|{#2}\right.\right\rangle}

\def\OL{\relax\ifmmode {\sf L}\else{\textsf L}\fi}
\def\OR{\relax\ifmmode {\sf R}\else{\textsf R}\fi}

\newcommand{\Id}{1\!\!1}
\newcommand{\ov}{\overline}
\begin{document}

\title[Weak commutation relations]{Weak commutation relations of unbounded operators and applications}

\author[F.Bagarello]{Fabio Bagarello}
\address{%
Dipartimento di Metodi e Modelli Matematici,
Universit\`a di Palermo,
90128 Palermo,
Italy}
\email{bagarell@unipa.it}

\author[A.Inoue]{Atsushi Inoue}
\address{Department of Applied Mathematics, Fukuoka University, Fukuoka 814-0180, Japan}
\email{a-inoue@fukuoka-u.ac.jp}

\author[C.Trapani]{Camillo Trapani}
\address{Dipartimento di Matematica e Informatica, Universit\`a di Palermo, I-90123 Palermo, Italy}
\email{trapani@unipa.it}
\maketitle

\begin{abstract}
Four possible definitions of the commutation relation $[S,T]=\Id$ of two closable unbounded operators $S,T$ are compared. The {\em weak} sense of this commutator is given in terms of the inner product of the Hilbert space $\H$ where the operators act. Some consequences on the existence of eigenvectors of two number-like operators are derived and the partial O*-algebra generated by $S,T$ is studied. Some applications are also considered.
\end{abstract}

\section{Introduction}\label{sec_0}
Giving a meaning to a formal commutation relation $[A,B]=C$, when $A,B,C$ are unbounded operators, can really be a touchy business. It is well known that even the fact that two operators commute can be understood in several different ways giving rise to deeply different conclusions. Nelson's example \cite{reedsimon} provides a beautiful instance where $A,B$ are essentially selfadjoint on a common invariant dense domain $\D$, they commute (i.e. $AB\xi=BA\xi$, for every $\xi \in \D$) but the corresponding spectral families (or, equivalently, the corresponding unitary groups they generate) do not commute. The situation becomes even more involved if we want to express a commutation relation $[A,B]=C$, with $C \neq 0$. Commutators for unbounded operators can easily be meaningless. However, in many concrete applications, like in Quantum Theories, they play a so relevant role to deserve a full-fledged mathematical consideration. In recent papers by one of us \cite{bagpb1}-\cite{bagrep}, generalizing the commutation relations for bosons $[a, a^\dagger]=\Id$, the more general case $[a,b]=\Id$, where $b$ is not the adjoint of $a$, has been considered and several interesting results on these {\em pseudo-bosons} have been derived, in particular for what concerns the existence and the behavior of bases of eigenvectors of two non self-adjoint operators.

Very likely, the most relevant link between {\em commutators} and {\em operators} is provided by the representation theory of infinite-dimensional Lie algebras. Let us in fact consider the Heisenberg Lie algebra $\mathfrak h$ generated by three elements $a, b, c$ whose Lie brackets are
$$ [a,b]=c, \quad [a,c]=[b,c]=0.$$
Representations of this Lie algebra have been extensively studied in the literature giving rise to a considerable amount of papers and monograph (we refer to \cite{barutracza, schmudgen} and references therein). One of the most relevant results in this topic is the Stone -  von Neumann theorem which states that every {\em integrable} representation of $\mathfrak h$ is unitarily equivalent to the Schr\"odinger representation given by the {\em annihilation} operator $a$ and by the {\em creation} operator $a^\dagger$, which satisfy the boson commutation relation $[a, a^\dagger]=\Id$. Thus studying operators $A,B$, with $B\neq A^\dagger$, which satisfy, in some sense, the commutation relation $[A,B]=\Id$ corresponds, finally, to analyzing certain {\em nonintegrable} representations of $\mathfrak h$. We will not however pursue this line.

The paper is organized as follows. In Section \ref{sec_1} we discuss and compare four different definitions of the commutation relation $[S,T]=\Id$, the operators $S,T$ being picked in the maximal {\em partial O*-algebra} $\LDH$ on a dense domain $\D$ of Hilbert space $\H$. To be more definite we recall that $ \L\ad(\D,\H) $ denotes
the set of all (closable) linear operators $X$ such that $ {D}(X) = {\D},\; {D}(X\x) \supseteq {\D}.$ The set $
\L\ad(\D,\H ) $ is a  \pa\
 with respect to the usual sum $X_1 + X_2 $,
the scalar multiplication $\lambda X$, the involution $ X \mapsto X\ad := X\x \up {\D}$ and the \emph{(weak)}
partial multiplication $X_1 \mult X_2 = {X_1}\ad\x X_2$, defined whenever $X_2$ is a weak right multiplier of
$X_1$ (we shall write $X_2 \in R^{\rm w}(X_1)$ or $X_1 \in L^{\rm w}(X_2)$), that is, whenever $ X_2 {\D} \subset
{\D}({X_1}\ad\x)$ and  $ X_1\x {\D} \subset {\D}(X_2\x)$. Any $\dagger-$invariant subspace $\MM$ of $\LDH$ stable under the weak multiplication is called a partial O*-algebra. By $\LD$ we will denote the maximal {\em O*-algebra} on $\D$ consisting of all element $X \in \LDH$ which leave, together with their adjoint, the domain $\D$ invariant. A *-subalgebra of $\LD$ is named an O*-algebra.

Among the possible definitions, we focus our attention, in particular, to the {\em weak} form of the commutation relation $[S,T]=\Id$, which is given in terms of the inner product of $\H$. This choice reveals to be sufficient for the analysis of the existence of eigenvectors and for the construction of intertwining operators considered in Section \ref{sec_consequences}.

The commutation relation $[S,T]=\Id$ (even when $T\neq S\ad$) plays a relevant role in concrete applications to quantum mechanics, and  have strong connections with what in the physical literature is called {\em pseudo-hermitian quantum mechanics}, (see \cite{bagrep} and references therein).

Section \ref{sec_4} is devoted to the construction of the partial O*-algebra generated by two operators $S,T$ satisfying the commutation relation $[S,T]=\Id$ in weak sense. The main outcome is that this partial O*-algebra decomposes into a regular part (a partial *-algebra of polynomials in the variables $S, T, S\ad, T\ad$) and a singular part whose control is more difficult. This is closely reminiscent of similar results discussed in \cite[Ch.3]{ait_book} for the case of commuting operators.

Finally, in Section \ref{sec_5} we discuss some more applicative aspects. In particular we derive two types of {\em uncertainty relations} for two operators $S,T$ satisfying the commutation relation $[S,T]=\Id$ in weak sense. These two uncertainty relations are proven to be independent by showing examples where one of them possesses a state of minimal uncertainty and the other not and viceversa.

\section{The commutation relation $[A,B]=\Id$} \label{sec_1}

Let $A,B$ be two closed operators with dense domain $D(A)$ and $D(B)$, respectively, in Hilbert space $\H$.
We begin with a discussion about the mathematical meaning of the formal commutation relation $[A,B]=\Id$.
Of course we want the identity $AB-BA=\Id$ to hold at least on a dense domain $\D$ of Hilbert space $\H$.
For this we require that there exists a dense subspace $\D$ of $\H$ such that
\begin{itemize}
\item[(D.1)] $\D \subset D(AB)\cap D(BA)$;
\item[(D.2)] $AB\xi-BA\xi = \xi, \quad \forall \xi \in \D$,
where, as usual, $D(AB)=\{ \xi \in D(B):\, B\xi \in D(A)\}$.
\end{itemize}
These two conditions do not provide any information about $D(A^*)$ and $D(B^*)$ apart from the fact that they are dense subspaces of $\H$. To get a better control on these two operators we suppose
\begin{itemize}
\item[(D.3)] $\D \subset D(A^*)\cap D(B^*)$.
\end{itemize}
By the previous assumptions it follows that the operators $S:=A\upharpoonright \D$ and $T:= B\upharpoonright D$ are elements of the partial *-algebra $\LDH$. It is readily checked that the operators $S$ and $T$ satisfy the equality
\begin{equation}\label{eq_starting} \ip{T\xi}{S\ad \eta}-\ip{S\xi}{T\ad \eta}= \ip{\xi}{\eta}, \quad \forall \xi, \eta \in \D.\end{equation}

It is then natural to consider \eqref{eq_starting} as a possible definition of the commutation relation $[S,T]=\Id_{\D}$ (where $\Id_{\D}$ denotes the identity operator of $\D$) for a couple of elements of $\LDH$.

\begin{lemma} \label{lemma_11}Let $S, T\in \LDH$ satisfy \eqref{eq_starting}. Then, if
 $S\mult T$ (resp. $T\mult S$)  is well defined then $T\mult S$ (resp. $S\mult T$) is well-defined and
$$ S\mult T - T\mult S =\Id_{\D}.$$
\end{lemma}
\begin{proof} Assume, for instance, that $S \mult T$ is well-defined.
Then, we have, for every $\xi, \eta \in \D$,
\begin{eqnarray*}
\ip{S\xi}{T\ad \eta}&=& \ip{T\xi}{S\ad \eta} - \ip{\xi}{\eta}\\
&=& \ip{(S\mult T)\xi}{\eta}- \ip{\xi}{\eta}.
\end{eqnarray*}
Hence, $S: \D \to D(T\das)$. Similarly one shows that $T\ad \eta \in D(S^*)$, for every $\eta \in \D$. Thus $T\mult S$ is also well defined and the equality $ S\mult T - T\mult S =\Id_{\D}$ follows immediately.
The proof under the assumption that $T\mult S$ is well-defined is analogous.
\end{proof}

\berem If $S, T$ are the restrictions of $A$, $B$ as above then we get a stronger result: if either $S\ad:\D \to D(T^*)$ or $T\ad: \D \to D(S^*)$, then
 $S\mult T$ and $T\mult S$ are both well-defined and
$ S\mult T - T\mult S =\Id_{\D}.$ This follows from the fact that (D.1) implies, evidently, that $T\D \subset D(A)\subset D(S\das)$ and $ S\D\subset D(B)\subset D(T\das)$\enrem

Before going forth we summarize some properties of semigroups of bounded operators.

\medskip
Let $t\mapsto V(t)$, $t\geq 0$, be a semigroup of bounded operators in Hilbert space. We recall that $V$ is weakly continuous if
$$ \lim_{t\to {t_0}}\ip{V(t)\xi}{\eta}= \ip{V(t_0)\xi}{\eta} , \quad \forall \xi, \eta \in \H.$$

A closed operator $X$ is the generator of $V(t)$ if
$$ D(X)= \left\{\xi\in \H;\exists\, \xi'\in \H: \lim_{t\to 0} \ip{\frac{V(t)-\Id}{t}\xi}{\eta}=\ip{\xi'}{\eta},\; \forall \eta \in \H \right\}$$
and
$$ X\xi = \xi', \quad \forall \xi \in D(X).$$

If $V(t)$ is a weakly continuous semigroup, then $V^*(t)$ is also a weakly continuous semigroup and if $X$ is the generator of $V(t)$, then $X^*$ is the generator of $V^*(t)$.

\berem In the previous discussion the assumption that $V(t)$ is a weakly continuous semigroup can be replaced with the assumption that $V(t)$ is strongly continuous since every weakly continuous semigroup is also strongly continuous and its generator is densely defined \cite[Ch.IX, Sect. 1]{yosida}.
\enrem

We will say that an operator $X_0\in \LDH$ is the $\D$-generator of a semigroup $V(t)$ if $V(t)$ is generated by some closed extension $X$ of $X_0$ such that $\overline{X_0}\subset X \subset X_0\das$.
The latter condition ensures us that if $X_0$ is the $\D$-generator of $V(t)$, then $X_0\ad$ is the $\D$-generator of $V^*(t)$.

\medskip
At the light of the previous discussion we give the following
\bedefi \label{defn_commutators}Let $S, T \in \LDH$. We say that
\begin{itemize}
\item[(CR.1)] the commutation relation $[S,T]=\Id_{\D}$ is satisfied (in $\LDH$) if, whenever $S\mult T$  is well-defined, $T\mult S$ is well-defined too and $S\mult T-T\mult S= \Id_{\D}$.
\item[(CR.2)] the commutation relation $[S,T]=\Id_{\D}$ is satisfied in {\em weak sense} if
$$ \ip{T\xi}{S\ad \eta}-\ip{S\xi}{T\ad \eta}= \ip{\xi}{\eta}, \quad \forall \xi, \eta \in \D.$$
\item[(CR.3)] the commutation relation $[S,T]=\Id_{\D}$ is satisfied in {\em quasi-strong sense} if $S$ is the    $\D$-generator of a weakly continuous semigroup of bounded operators $V_S(\alpha)$ and
    $$\ip{V_S(\alpha)T\xi}{\eta}-\ip{V_S(\alpha)\xi}{T\ad \eta}=\alpha\ip{V_S(\alpha)\xi}{\eta}, \quad \forall \xi, \eta \in \D, \forall \alpha \geq 0.$$
\item[(CR.4)] the commutation relation $[S,T]=\Id_{\D}$ is satisfied in {\em strong sense} if $S$ and $T$ are $\D$-generators of weakly continuous semigroups of bounded operators $V_S(\alpha), V_T(\beta)$, respectively, satisfying the generalized Weyl commutation relation
    $$ V_S(\alpha)V_T(\beta) = e^{\alpha\beta} V_T(\beta)V_S(\alpha), \quad \forall \alpha, \beta \geq 0.$$
\end{itemize}
\findefi

\berem If (CR.3) holds, then one can easily prove that , for every $\alpha \geq 0$, $V_S(\alpha):\D\to D(T\das)$ and
$V_S(\alpha)^*:\D\to D(T^*)$. Hence $V_S(\alpha)\mult T$ and $T \mult V_S(\alpha)$ are both well-defined (we use the same notation for $V_S(\alpha)$ and for its restriction to $\D$) and we have
$$ V_S(\alpha)\mult T - T \mult V_S(\alpha)= \alpha V_S(\alpha).$$

\enrem

\begin{prop}Let $S, T \in \LDH$. The following implications hold: {\rm (CR.4)} $\Rightarrow$ {\rm (CR.3)} $\Rightarrow$ {\rm(CR.2)} $\Rightarrow$ {\rm (CR.1)}.

\end{prop}
\begin{proof} Assume that (CR.4) holds. Then if $\xi, \eta \in \D$ we have
\begin{eqnarray*}
\ip{V_S(\alpha)T\xi}{\eta}&=& \ip{T\xi}{V_S^*(\alpha) \eta}\\
&=& \ip{\lim_{\beta \to 0}\frac{V_T(\beta)-\Id_{\D}}{\beta}\xi}{V_S^*(\alpha) \eta}\\
 &=&  \lim_{\beta \to 0}\ip{V_S(\alpha)\frac{V_T(\beta)-\Id_{\D}}{\beta}\xi}{ \eta}
\\&=&  \lim_{\beta \to 0}\ip{\frac{e^{\alpha\beta}V_T(\beta)V_S(\alpha)-V_S(\alpha)\xi}{\beta}}{\eta}\\
&=&  \lim_{\beta \to 0}\ip{V_S(\alpha)\xi}{ \frac{e^{\alpha\beta}V_T^*(\beta)-\Id}{\beta}\eta}\\
&=& \lim_{\beta \to 0}\ip{V_S(\alpha)\xi}{ \frac{e^{\alpha\beta}V_T^*(\beta)-e^{\alpha\beta}\Id +e^{\alpha\beta}\Id-\Id}{\beta}\eta}\\
&=& \ip{V_S(\alpha)\xi}{T\ad\eta}+\alpha\ip{V_S(\alpha)\xi}{\eta}.
\end{eqnarray*}
Thus (CR.3) holds.

Now assume that (CR.3) holds. Then if $\xi, \eta \in \D$ we have
\begin{eqnarray*}
\ip{T\xi}{S\ad\eta}&=& \lim_{\alpha \to 0}\ip{T\xi}{\frac{V_S^*(\alpha)-\Id}{\alpha} \eta}\\
&=& \lim_{\alpha \to 0}\ip{\frac{V_S(\alpha)-\Id}{\alpha}T\xi}{\eta}\\
&=&\lim_{\alpha \to 0} \left\{\ip{\frac{V_S(\alpha)-\Id}{\alpha}\xi}{T\ad\eta}+\ip{V_S(\alpha)\xi}{\eta}  \right\}\\
&=&\ip{S\xi}{T\ad \eta}+\ip{\xi}{\eta}.
\end{eqnarray*}
(CR.2)$\Rightarrow$(CR.1): this is nothing but Lemma \ref{lemma_11}.

\end{proof}

\berem The implications in the other direction do not hold, in general, also in the case where stronger assumptions on the operators $S,T$ or on the semigroups they generate are made. For instance, there exist two essentially selfadjoint operators $P,Q$ defined on a common invariant dense domain $\D$ such that $PQ\xi -QP\xi =-i\xi$, for $\xi \in \D$, but the unitary groups $U_P(t), U_Q(s)$ generated by $\overline{P}$, $\overline{Q}$ do not satisfy the Weyl commutation relation $U_P(t)U_Q(s)=e^{its}U_Q(s)U_P(t)$, $s,t \in {\mb R}$, see \cite{fuglede} or \cite[VIII.5]{reedsimon}. For a complete analysis of the canonical commutation relations (for symmetric operators) we refer to \cite{schmud_hei1, schmud_hei2}.\enrem

\berem \label{rem_comm_adj}It is easily seen that if $[S,T]=\Id_{\D}$ in anyone of the senses (CR.1), (CR.2) or (CR.4) of Definition \ref{defn_commutators}, then $[T\ad, S\ad]=\Id_{\D}$ in the same sense.\enrem

\section{Some consequences of the weak commutation relation}\label{sec_consequences}
\subsection{Existence of eigenvectors}\label{subsec_eigenvectors}
Let $S,T \in \LDH$ satisfy the commutation relation $[S,T]=\Id_{\D}$ in weak sense. Assume that there exists a vector $0\neq \xi_0\in \D$ such that $S\xi_0=0$. Then
\begin{equation}\label{two}\ip{T\xi_0}{S\ad \eta}= \ip{\xi_0}{\eta}, \quad \forall \eta \in \D.\end{equation}
This implies that $T\xi_0\neq 0$ (otherwise, $\xi_0=0$) and
$$ T\xi_0 \in D(S\das), \quad S\das T\xi_0 =\xi_0 \in \D.$$
Thus $T$ can be applied once more and we get
\begin{equation}\label{three} (TS\das)T\xi_0 =T(S\das T\xi_0) = T\xi_0.  \end{equation}
Hence $T\xi_0$ is an eigenvector of $TS\das$ with eigenvalue $1$.

\berem From \eqref{three} it follows
$$\ip{T\xi_0}{S\ad \eta}= \ip{(TS\das)T\xi_0}{S\ad \eta}=\ip{\xi_0}{\eta}, \quad \forall \eta \in \D.$$
Hence $(TS\das)T\xi_0 \in D(S\das)$ (which is obvious) and
$$ S\das(T S\das)T\xi_0 = \xi_0 \Rightarrow (TS\das)^2 T\xi_0 = T\xi_0. $$
By induction one can prove that $T\xi_0 \in D((TS\das)^n)$, for all $n\in {\mb N}$, and $(TS\das)^n T\xi_0 =T\xi_0$.
\enrem

Let us now assume that $T\xi_0 \in \D$. Then we have
$$ \ip{T(T\xi_0)}{S\ad \eta}- \ip{S(T\xi_0)}{T\ad \eta}= \ip{T\xi_0}{\eta}, \quad \forall \eta \in \D.$$
Whence we obtain
$$ \ip{T^2\xi_0}{S\ad \eta}-\ip{\xi_0}{T\ad \eta} = \ip{T\xi_0}{\eta}, \quad \forall \eta \in \D$$
since $ST\xi_0 = S\das T\xi_0=\xi_0$. Finally
\begin{equation}\label{four}
\ip{T^2\xi_0}{S\ad \eta} = 2\ip{T\xi_0}{\eta}, \quad \forall \eta \in \D.\end{equation}
This implies that $T^2\xi_0 \in D(S\das)$ and $S\das T^2 \xi_0 =2 T\xi_0$.
Therefore $$(TS\das) T^2 \xi_0 = 2 T^2 \xi_0.$$
Also $T^2\xi_0 \neq 0$, since, otherwise, $T\xi_0=0$ by \eqref{four}.
Iterating this procedure we conclude that
\beprop \label{prop_eigenvectors} Let $S,T \in \LDH$ satisfy the commutation relation $[S,T]=\Id_{\D}$ in weak sense. Assume that there exists a nonzero vector $\xi_0\in \D$ such that $S\xi_0=0$ and the vectors $T\xi_0, T^2\xi_0,\ldots T^{n-1}\xi_0$ all belong to $\D$. Then
\begin{itemize}
\item[(i)] $T^n\xi_0$ is an eigenvector of $TS\das$ with eigenvalue $n$;
\item[(ii)] $T^{n-1}\xi_0$ is eigenvector of $S\das T$ with eigenvalue $n$.
\end{itemize}
\enprop
\begin{proof} For proving (i), we proceed by induction. For $n=1$ the statement is true. Assume it is true for $n-1$. Then we have,
\begin{equation}\label{five} \ip{T(T^{n-1}\xi_0)}{S\ad \eta}- \ip{S(T^{n-1}\xi_0)}{T\ad \eta}= \ip{T^{n-1}\xi_0}{\eta}, \quad \forall \eta \in \D.\end{equation}
By the assumption, $T^{n-1}\xi_0 \in D(TS\das)$ and $(TS)T^{n-1}\xi_0 =(TS\das)T^{n-1}\xi_0 = (n-1)T^{n-1}\xi_0$; thus by \eqref{five}, we obtain
$$ \ip{T^n\xi_0}{S\ad \eta}- (n-1)\ip{T^{n-1}\xi_0}{\eta}= \ip{T^{n-1}\xi_0}{\eta}, \quad \forall \eta \in \D.$$
This implies that $T^n\xi_0 \in D(S\das)$ and $S\das T^n\xi_0 = n T^{n-1}\xi_0$. The latter equality shows (ii) and implies that  $(TS\das)T^n\xi_0 = nT^n \xi_0$.
\end{proof}

\berem If $\xi_0\neq 0$, the vectors $\xi_0, T\xi_0, T^2\xi_0,\ldots T^{n}\xi_0$ are linearly independent, being eigenvectors corresponding to different eigenvalues.
\enrem

\beprop Let $S,T \in \LDH$ satisfy the commutation relation $[S,T]=\Id_{\D}$ in weak sense. Let $\xi \in \D$ and assume that $T^k \xi \in \D$ for $k\leq n$, $n \in {\mb N}\cup \{\infty\}$.
Then $S\xi \in D((T\das)^k)$, $k\leq n$ and
$$ ST^k\xi - (T\das)^k S\xi= kT^{k-1}\xi, \quad k\leq n.$$
\enprop

\begin{proof} Again, we use induction on $k$. For $k=1$, the statement follows directly from the weak commutation rule. Assume the statement is true for $k-1$.
Then
\begin{eqnarray*} ST^k \xi &=& ST^{k-1}(T\xi)\\
&=& (k-1)T^{k-2}(T\xi)+ (T\das)^{k-1}ST\xi \\
&=& (k-1)T^{k-1}\xi + (T\das)^{k-1}(T\das S\xi +\xi)\\
&=& kT^{k-1}\xi + (T\das)^k S\xi.
\end{eqnarray*}
\end{proof}
\berem In particular if $S\xi=0$ then $ST^k\xi =kT^{k-1}\xi$.
\enrem

If the assumptions of Proposition \ref{prop_eigenvectors} are satisfied one may have that the largest $n$ for which $T^n\xi_0\in \D$ is finite or infinite.
As we have seen the point spectrum $\sigma_p(TS\das)$ contains all natural numbers up to $n$. Let us denote by $\N_0$ the subspace of $\D$ spanned by $\{\xi_o, T\xi_0, \ldots T^n\xi_0\}$ and by $\N:=\overline{\N_0}$ its closure in $\H$.
Clearly $TS\das$ leaves $\N_0$ invariant. The restriction of $TS\das$ to $\N_0$, denoted by $(TS\das)_0$, behaves in quite regular way. We have indeed

\beprop In the assumptions of Proposition \ref{prop_eigenvectors}, the point spectrum $\sigma_p((TS\das)_0)$ of the operator $(TS\das)_0$, the restriction of $TS\das$ to $\N_0$, consists exactly of the set $\{0,1,\ldots, n\}$, where $n \in {\mb N}\cup \{\infty\}$ is the largest natural number such that the vectors $T\xi_0, T^2\xi_0,\ldots T^{n-1}\xi_0$ all belong to $\D$. Each eigenvalue is simple (in $\N_0$).
\enprop
\begin{proof} Let $\lambda_0$ be an eigenvalue of $(TS\das)_0$ and $0\neq \psi_0=\sum_{k=0}^n \mu_k T^k \xi_0$ a corresponding eigenvector. Then

$$ TS\das \psi_0 -\lambda_0 \psi_0 = \sum_{k=0}^n \mu_k k T^k \xi_0 - \lambda_0 \sum_{k=0}^n \mu_k T^k \xi_0= \sum_{k=0}^n \mu_k (k-\lambda_0)T^k \xi_0=0.$$
The linear independence of the vectors $\{T^k\xi_0, k=0,1,\dots n\}$ then implies that either $\mu_k=0$ which is absurd, or that there exists $k \in \{0,1,\ldots, n\}$ such that $\mu_j=0$ for $j\neq k$ and $\lambda_0=k$.

\end{proof}

\beex \label{ex_weights}
Let us consider the Hilbert space $L^2({\mb R}, wdx)$ where the weight $w$ is a positive continuously differentiable function with the properties
\begin{itemize}
\item $\lim_{|x|\to +\infty} w(x)=0$;
\item $\int_{\mb R}w(x)dx <\infty$.
\end{itemize}
Let
$$D({\sf p})=\left\{ f \in L^2({\mb R}, wdx):\, \exists g \in L^2({\mb R}, wdx), f(x)= \int_{-\infty}^x g(t) dt\right\}.$$
For shortness, we adopt the notation $f'(x)=g(x)$, for $f \in D({\sf p})$.
$$D({\sf q})=\{ f\in L^2({\mb R}, wdx): xf(x) \in L^2({\mb R}, wdx)\}.$$

Put $\D= D({\sf q})\cap D({\sf p})$. Then both the operators $S, T$ defined by
$$ (Sf)(x) =f'(x), \qquad  (Tf)(x)=xf(x), \quad f \in \D$$
map $\D$ into $L^2({\mb R}, wdx)$.

The operator $T$ is symmetric in $\D$. As for $S$, we have (formally) for $f\in \D$ and $ g \in L^2({\mb R}, wdx)$,
\begin{eqnarray*}
\ip{Sf}{g}&=& \int_{\mb R}f'(x)\overline{g(x)} w(x) dx =-\int_{\mb R}f(x)\frac{d}{dx}(\overline{g(x)}w(x))dx \\
&=& -\int_{\mb R} f(x)\overline{g'(x)}w(x)dx - \int_{\mb R} f(x) \overline{g(x)} w'(x) dx \\
&=& -\int_{\mb R} f(x)\overline{g'(x)}w(x)dx - \int_{\mb R} f(x) \overline{g(x)} \frac{w'(x)}{w(x)}\, w(x) dx.
\end{eqnarray*}
Thus, $g \in D(S^*) $ if, and only if $g \in D({\sf p})$ and $g\frac{w'}{w} \in L^2({\mb R}, wdx)$. In this case
$$ (S^* g)(x)= -g'(x) - g(x)\frac{w'(x)}{w(x)}.$$
Hence $S \in \LDH$, with $\H=L^2({\mb R}, wdx)$ if, and only if, $$\int_{\mb R} |g(x)|^2 \frac{|w'(x)|^2}{w(x)} dx<\infty.$$
This is certainly satisfied if, for instance, $w'/w$ is a bounded function on ${\mb R}$.
Now, for $f,g\in \D$, we have
\begin{eqnarray*}
\ip{Tf}{S\ad g} -\ip{Sf}{T\ad g}&=-& \int_{\mb R}x f\,\left(\overline{g'}+\overline{g}\frac{w'}{w}\right)wdx - \int_{\mb R} f'\, x\overline{g} wdx  \\
&=& -\int_{\mb R} x(f' \overline{g}w +f\overline{g'}w +f \overline{g}w')dx \\
&=& -\int_{\mb R} x (f\overline{g}w)' dx = \int_{\mb R} f\overline{g}w dx = \ip{f}{g}.
\end{eqnarray*}
Hence, the commutation relation $[S,T]=\Id_{\D}$ is satisfied in weak sense.

The function $u_0(x)=1$, for every $x \in {\mb R}$, is clearly in the kernel of $S$ for every function $w$ satisfying the assumptions made so far.

Now we make some particular choice of $w$.

\medskip
Let us consider $w(x)=w_\alpha(x)= (1+x^4)^{-\alpha}$, $\alpha>\frac{3}{4}$. It is easily seen that $w_\alpha(x)$ satisfies all the conditions we have required (for instance $w'_\alpha/w_\alpha$ is bounded). The function $u_0(x)=1$, which belongs to  $L^2({\mb R}, w_\alpha dx)$ for any $\alpha>\frac{3}{4}$, satisfies $Su=0$ and the largest $n$ for which $T^n u_0$ belongs to $\D$ satisfies $n<2\alpha-\frac{3}{2}$. Hence, the dimension of the corresponding subspace $\N_0$ is $\left[2\alpha-\frac{3}{2}\right]+1$.

\medskip
Let us now take $w(x)= e^{-x^2/2}$ and $\D$ the subspace consisting of all polynomials in $x$. In this case the functions $u_k(x)= x^k$, $k=1,2,\ldots,$ belong to $\D$ and they satisfy $TS\das u_k= ku_k$ for every $k \in {\mb N}$. The subspace $\N_0$ coincides in this case with $\D$. One can readily check that every complex number $\lambda$ with $\Re \lambda >-\frac{1}{2}$ is an eigenvalue of $TS\das$; but the corresponding eigenvector is in $\D$ if and only if  $\Re\lambda$ is  a natural number. \enex

\medskip As we have seen with the previous examples the subspace $\N_0$ spanned by $\{T^k\xi_0, k\in {\mb N}\}$ can be finite dimensional. Thus $N:= (TS\das)_0$ is a bounded symmetric operator on $\N_0=\N \cong {\mb C}^n$, having the numbers $0,1,\ldots, n$ as eigenvalues. Hence $N$ is positive and thus there exists an operator $C\in \B(\N)$ such that $N=C\ad C $. None of the possible solutions of this operator equation can, however, satisfy the commutation relation $[C,C\ad]=\Id$, due to the Wiener -Wielandt - von Neumann theorem. If $\N_0$ is infinite dimensional then $N$ may fail to be symmetric, as the last case in Example \ref{ex_weights} shows.

\subsection{Existence of intertwining operators}\label{subsec_intertwining}
As noticed in Remark \ref{rem_comm_adj}, if $[S,T]=\Id_{\D}$ in weak sense, then $[T^\dagger,S^\dagger]=\Id_\D$ in weak sense, too. Hence it is worth considering the situation where the assumptions of Proposition \ref{prop_eigenvectors} hold not only for the pair $S, T$ but also for the pair $T\ad, S\ad$. This means that we assume that there exists also
a nonzero vector $\eta_0\in \D$ such that $T\ad\eta_0=0$ and the vectors $S^\dagger\eta_0$, $(S^\dagger)^2\eta_0,\ldots$ $(S^\dagger)^{m-1}\eta_0$ all belong to $\D$, then Proposition \ref{prop_eigenvectors} gives that
\begin{itemize}
\item[(i)] $(S^\dagger)^m\eta_0$ is an eigenvector of $S\ad T^*$ with eigenvalue $m$;
\item[(ii)] $(S^\dagger)^{m-1}\eta_0$ is eigenvector of $T^*S\ad$ with eigenvalue $m$.
\end{itemize}

Let $m\in {\mb N}\cup \{\infty\}$ be the largest natural number such that the vectors  $S^\dagger\eta_0$, $(S^\dagger)^2\eta_0,\ldots,$ $(S\ad)^{m-1}\eta_0$ all belong to $\D$ and $\M_0$ the linear span of these vectors. Then the point spectrum $\sigma_p((S\ad T^*)_0)$, of the operator $(S\ad T^*)_0:= S\ad T^*\upharpoonright\M_0$, consists, as before,  of the numbers $\{0,1,\ldots, m\}$.

One may wonder if any relation between the two numbers $n$ and $m$ can be established. The answer is negative, in general. Indeed, the operators $S,T$ considered in the second case of Example \ref{ex_weights} provide an instance where $n=\infty$ and $m=0$. Thus it is apparently impossible to find a relationship between $\N_0$ and $\M_0$ without additional assumptions.

Let us now call $\xi_k:=\frac{1}{\sqrt{k!}}T^k\xi_0$, $k=1, \ldots, n$ and $\eta_r:=\frac{1}{\sqrt{r!}}(S^\dagger)^r\eta_0$, $r=1, \ldots,m$. It is always possible to choose the normalization of $\xi_0$ and $\eta_0$ in such a way that $\ip{\xi_0}{\eta_0}=1$.We put $\F_\xi:=\{\xi_k; k=1, \ldots, n \}$ $\F_\eta:=\{\eta_r; r=1, \ldots, m \}$. Then we have

\begin{lemma} The sets $\F_\xi$ and $\F_\eta$ are biorthogonal:
$$
\ip{\xi_i}{\eta_j}=\delta_{i,j},
$$
for all $i=1, \ldots, n$ and $j=1, \ldots, m$.

\end{lemma}

The proof of this Lemma easily follows from the fact that $\xi_i$ and $\eta_j$ are eigenvectors of the two (non self-adjoint) operators, $TS\das$ and $S\ad T^*$, having the property $(TS\das)^* \supseteq \overline{S\ad}T^* \supseteq S\ad T^* $.

Assume now that $n=m=\infty$. Thus the subspaces $\N_0$ and $\M_0$ are both infinite dimensional. Then we can define two operators which obey interesting intertwining relations. More in detail, let us define $K_\xi$ via its action on the basis $\F_\eta$: $K_\xi(\eta_j)=\xi_j$, $j\in {\mb N}$. We can also introduce a second operator $K_\eta$ via its action on the second basis constructed above, $\F_\xi$: $K_\eta(\xi_j)=\eta_j$, $j\in {\mb N}$. Both the operators $K_\xi$ and $K_\eta$ are then extended by linearity to $\M_0$ and $\N_0$, respectively. It is clear that one is the inverse of the other: $K_\eta=K_\xi^{-1}$, but in general neither $K_\xi$ nor $K_\eta$ are bounded. A direct computation shows that they obey the following intertwining relations:
$$
K_\eta\left(TS\das\right)\phi =\left(S^\dagger T^*\right)K_\eta\phi,\; \forall \phi \in \M_0;
$$
$$
K_\xi\left(S^\dagger T^*\right)\psi=\left(TS\das\right)K_\xi\psi,\; \forall \psi \in \N_0.
$$
In particular, if $\N_0=\M_0=\H$ and $K_\xi$ and $K_\eta$ are bounded, it is possible to show that $\F_\xi$ and $\F_\eta$ are Riesz bases of $\H$ and an orthonormal basis $\E=\{e_j\}$ can be defined by, for instance, $e_j=K_\eta^{1/2}\xi_j$ (see \cite{bagpb1}).

\berem The analysis made so far started with the assumption that $[S,T]=\Id_D$ in {\em weak } sense.
It is then natural to pose the question which improvements are obtained if we consider, instead, that (CR.3) or (CR.4) is fulfilled. More precisely, assume that (CR.3) is satisfied and that there exists a nonzero $\xi_0\in \D$ such that, $V_S(\alpha)\xi_0 = \xi_0$, for every $\alpha \geq 0$ (i.e., an {\em invariant} vector for $V_S(\alpha)$). This implies, as it is easily seen that $S\xi_0=0$ and so everything goes through in the very same way as before.
\enrem

It is worth mentioning here that the properties discussed in Sections \ref{subsec_eigenvectors} and \ref{subsec_intertwining} depend in an essential way on the nonzero vectors $\xi_0$ and $\eta_0$ satisfying, respectively, $S\xi_0=0$ and $T\ad\eta_0=0$: choosing different elements in the kernels of $S$ and $T\ad$ may give drastically different results.

\section{The partial O*-algebra generated by $S$ and $T$}\label{sec_4}

In this Section we will describe the partial O*-algebra generated by $S,T\in \LDH$ such that $[S,T]=\Id_{\D}$ in weak sense, starting from observing that in the case where $S,T \in \LD$ and $ST-TS=\Id_D$ the O*-algebra generated by them is very {\em regular} in the sense that it consists only in polynomials in the variables $S^kT^h$, $(S\ad)^r(T\ad)^s$, with $k,h,r,s\in {\mb N}$. As we shall see below, the situation becomes more involved if $S$ and $T$ do not leave the domain $\D$ invariant. For an analogous analysis for commuting operators we refer the reader to \cite[Sect. 3.3]{ait_book}.

\medskip
Let $ S,T\in \LDH$ satisfy the commutation relation $[S,T]=\Id_{\D}$ in weak sense. Assume that $S^*=\overline{S\ad}$, $T^*= \overline{T\ad}$. This clearly implies that $(\overline{S\ad})^k = (S^*)^k$, $(\overline{T\ad})^k = (T^*)^k$. We want to describe the partial O*-algebra $\N[S,T]$ generated by $S,T$ and study its structure. We begin with defining the {\em power length} of $S,S\ad, T$ and $T\ad$.

Let $m_0$ be the largest of the numbers $k \in {\mb N}\cup \{\infty\}$ satisfying $ \D \subset D(\overline{S}^k) \cap D(\overline{S\ad}^k)$.
Suppose that
\begin{equation}\label{eq_one}
\D \subset D(\overline{T}S)\cap D(T^*S^*).
\end{equation}

Then, if $\eta \in \D$,
$$ \ip{S\ad\xi}{T\eta} = \ip{T\ad \xi}{S\eta}+\ip{\xi}{\eta} = \ip{\xi}{\overline{T}S\eta} + \ip{\xi}{\eta}, \quad \forall \xi \in \D.$$

Hence,
\begin{equation}\label{eq_two}\eta \in D(S\das T)=D(\overline{S}T)\; \mbox{ and } \;\overline{S}T\eta = \overline{T}S\eta + \eta. \end{equation}

Similarly, from the equality

$$ \ip{S\xi}{T\ad \eta} = \ip{T \xi}{S\ad\eta}-\ip{\xi}{\eta} = \ip{\xi}{T^*S\ad\eta} - \ip{\xi}{\eta}, \quad \forall \xi \in \D,$$
we obtain

\begin{equation}\label{eq_three}\eta \in D(S^* T)\; \mbox{ and } \;S^*T\ad\eta = T^*S\ad\eta - \eta. \end{equation}

Suppose \eqref{eq_one} holds and that
\begin{equation}\label{eq_four}
\D \subset D(\overline{T}\,\overline{S}^2)\cap D(T^*{S^*}^2).
\end{equation}

Then, if $\eta \in \D$, for every $\xi \in \D$, we get
\begin{align*} \ip{S\ad\xi}{\overline{S}T\eta} &=\ip{S\ad\xi}{\overline{T}S\eta+ \eta} & \mbox{by \eqref{eq_two}}\\
& =  \ip{S\ad\xi}{\overline{T}S\eta} + \ip{S\ad\xi}{\eta} & {}\\
&=\ip{T^*S\ad\xi}{S\eta} + \ip{\xi}{S\eta} & \mbox{by \eqref{eq_one}}\\
&= \ip{S^*T\ad\xi +\xi}{S\eta} + \ip{\xi}{S\eta} & \mbox{by \eqref{eq_two}}\\
&=\ip{\xi }{\overline{T}\,\overline{S}^2\eta} + \ip{\xi}{2S\eta} & \mbox{by \eqref{eq_four}.}
\end{align*}

Hence,
\begin{equation}\label{eq_five}\eta \in D(\overline{S}^2T)\; \mbox{ and } \;\overline{S}^2T\eta = \overline{T}\overline{S}^2\eta + 2S\eta. \end{equation}

Similarly, one can prove

\begin{equation}\label{eq_six}\eta \in D({S^*}^2 T\ad)\; \mbox{ and } \;{S^*}^2T\ad \eta = T^*{S^*}^2\eta -2 S\ad \eta. \end{equation}
Repeating the above argument, we get that, if

\begin{equation}\label{eq_seven}
\D \subset \bigcap_{k=1}^n D(\overline{T}\,\overline{S}^k)\cap D(T^*{S^*}^k),
\end{equation}
then

\begin{equation}\label{eq_eight}
\D \subset \bigcap_{k=1}^n D(\overline{S}^kT)\cap D({S^*}^kT\ad)
\end{equation}
and

\begin{equation}\label{eq_nine}
\overline{S}^kT \xi = \overline{T}\,\overline{S}^k \xi+ k \overline{S}^{k-1}\xi, \quad \forall \xi \in \D.
\end{equation}

Thus we define $m_1$ as the largest of the numbers  $n \in {\mb N}\cup \{\infty\}$ such that
$$ \D \subset \bigcap_{k=1}^n D(\overline{T}\,\overline{S}^k)\cap D(T^*{S^*}^k).$$
Then, $m_0\geq m_1$,
$$ \D \subset \bigcap_{k=1}^n D(\overline{S}^kT)\cap D({S^*}^kT\ad)$$

and

\begin{align}
& \overline{S}^kT \xi = \overline{T}\,\overline{S}^k \xi+ k \overline{S}^{k-1}\xi, \\
& {S^* }^kT\ad \xi = {T}^*\,{S^*}^k \xi- k {S^*}^{k-1}\xi, \quad \forall \xi \in \D \mbox{ and } 0\leq k \leq m_1. \nonumber
\end{align}

Now, suppose that
$$ D\subset D(\ov{T}^2)\cap D({T^*}^2) \cap D(\ov{T}^2 S) \cap D({T^*}^2 S\ad).$$
Changing $S$ and $T$ in \eqref{eq_five} and \eqref{eq_six}, we get
$$D \subset D(\ov{S}\,\ov{T}^2) \cap D(S^*{T^*}^2)$$
and
\begin{align}\label{eq_ten}
& \ov{S}\,\ov{T}^2 \xi = \ov{T}^2 S \xi +2 T\xi \\
& S^*{T^*}^2 \xi = {T^*}^2 S\ad \xi - 2T\ad \xi, \quad \forall \xi \in \D. \nonumber
\end{align}

Now suppose that
$$\D \subset \bigcap_{k=0}^2 D(\overline{T}^2\,\overline{S}^k)\cap D({T^*}^2{S^*}^k).$$
Then, if $\eta \in \D$, for every $\xi \in \D$, we get
\begin{align*} \ip{S\ad\xi}{\overline{S}\,\ov{T}^2\eta} &=\ip{S\ad\xi}{\overline{T}^2S\eta+ 2T\eta} & \mbox{by \eqref{eq_ten}}\\
&=\ip{{T^*}^2S\ad\xi}{S\eta} + \ip{\xi}{2\ov{S}T\eta} & {}\\
&= \ip{S^*{T^*}^2\xi +2T\ad\xi}{S\eta} + \ip{\xi}{2\ov{S}T\eta} & \mbox{by \eqref{eq_ten}}\\
&=\ip{\xi }{\overline{T}^2\,\overline{S}^2\eta} + \ip{\xi}{2\ov{T}S\eta}+\ip{\xi}{2\ov{S}T\eta} & {} \\
&=\ip{\xi}{\ov{T}^2\,\ov{S}^2\eta +4 \ov{T}S\eta +2\eta} &  \mbox{by \eqref{eq_two}.}
\end{align*}
Hence, $$\eta \in D(\ov{S}^2 \,\ov{T}^2) \mbox{ and } \ov{S}^2 \,\ov{T}^2\eta = \ov{T}^2\,\ov{S}^2 \eta +4 \ov{T}S\eta +2\eta.$$
Similarly, $$\eta \in D({S^*}^2 \,{T^*}^2)\mbox{ and } {S^*}^2 \,{T^*}^2\eta = {T^*}^2\,{S^*}^2 \eta -4 {T^*} {S^*}\eta -2\eta.$$

\medskip Repeating the above argument, we get that, if
$$\D \subset \bigcap_{k=0}^n D(\overline{T}^2\,\overline{S}^k)\cap D({T^*}^2{S^*}^k),$$ then

\begin{equation}\label{eq_dom0}\D \subset \bigcap_{k=0}^n D(\overline{S}^k\,\overline{T}^2)\cap D({S^*}^k{T^*}^2)\end{equation} and, for every $\eta \in \D$,

\begin{align} &\label{eq_dom1} \ov{S}^k \,\ov{T}^2\eta = \ov{T}^2\,\ov{S}^k \eta +\sum_{k=0}^2\sum_{\ell=0}^{n-1}\alpha_{k\ell} \ov{T}^k\ov{S}^\ell  \eta\\
 &{} \nonumber\\
& \label{eq_dom2}{S^*}^k \,{T^*}^2\eta = {T^*}^2\,{S^*}^k \eta -\sum_{k=0}^2\sum_{\ell=0}^{n-1}\alpha_{k\ell} {T^*}^k{S^*}^\ell  \eta ,\end{align}
where the $\alpha_{k\ell}$'s are integer numbers.

Thus we may define $m_2$ as the largest $n \in {\mb N}\cup\{0\}\cup\{\infty\}$ satisfying
$$\D \subset \bigcap_{k=0}^n D(\overline{T}^2\,\overline{S}^k)\cap D({T^*}^2{S^*}^k).$$
Clearly, $m_0\geq m_1\geq m_2$ and \eqref{eq_dom0}, \eqref{eq_dom1} and \eqref{eq_dom2} hold for every $k=0,\ldots, m_2$.
In the very same way, if all powers up to $n_0$ of $\ov{T}$ are defined, repeating the above procedure with the appropriate assumptions on the domains, we can define numbers
$$ m_0 \geq m_1 \geq m_2 \geq \cdots \geq m_{n_0}$$
and for every $r=1,\ldots, n_0$ we will get, for every $\eta \in \D$,
\begin{align} &\label{eq_dom3} \ov{S}^k \,\ov{T}^r\eta = \ov{T}^r\,\ov{S}^k \eta +\sum_{k=0}^r\sum_{\ell=0}^{m_r-1}\alpha_{k\ell} \ov{T}^k\ov{S}^\ell  \eta\\
 &{} \nonumber\\
& \label{eq_dom4}{S^*}^k \,{T^*}^r\eta = {T^*}^r\,{S^*}^k \eta -\sum_{k=0}^r\sum_{\ell=0}^{m_r-1}\alpha_{k\ell} {T^*}^k{S^*}^\ell  \eta .\end{align}

\medskip
>From the assumptions $S^*=\ov{S\ad}$ , $T^*=\ov{T\ad}$, it follows that ${S^*}^k= \ov{S\ad}^k$, ${T^*}^k = \ov{T\ad}^k$. Hence we put $X\pow{k}:= \ov{X}^k \upharpoonright \D$, with $X\in \{S, S\ad, T, T\ad\}$.

By \eqref{eq_dom3} and \eqref{eq_dom4}, we can define the {\em regular part} ${\mc R}[S,T]$ of the partial O*-algebra ${\mc M}[S,T]$ generated by $S,T$ as
\begin{align*}
{\mc R}[S,T] &= \left\{ p\pow{0}(S,T):= \sum_{k=0}^{m_0}(\alpha_{k0}S\pow{k} + \beta_{k0}{S\ad}\pow{k})\right. \\
&+  \left. \sum_{k=0}^{m_1}(\alpha_{k1}S\pow{k}T\pow{1} + \beta_{k1}{S\ad}\pow{k}{T\ad}\pow{1})\right.\\
& + \cdots \cdots \\
&+ \left. \sum_{k=0}^{m_{n_0}}(\alpha_{kn_0}S\pow{k}T\pow{n_0} + \beta_{kn_0}{S\ad}\pow{k}{T\ad}\pow{n_0});
\;\alpha_{kj}, \, \beta_{kj} \in {\mb C}
\right\}
\end{align*}

The {\em singular part} ${\mc S}[S,T]$ of ${\mc M}[S,T]$ is defined as  $${\mc S}[S,T]:= {\mc M}[S,T] \setminus {\mc R}[S,T].$$
In order to describe the singular part ${\mc S}[S,T]$ we define the following sets.

\begin{align*}
\Sigma_1[S,T]&= \mbox{ linear span }\{ p\pow{1}[S,T]= p_1\pow{0}[S,T] \mult p_2\pow{0}[S,T];\\ &\, p_1\pow{0}[S,T], p_2\pow{0}[S,T]\in {\mc R}[S,T], p_1\pow{0}[S,T]\in L^w (p_2\pow{0}[S,T])\}\\
\Sigma_2[S,T]&= \mbox{ linear span }\{ p\pow{2}[S,T]= p_1\pow{1}[S,T] \mult p_2\pow{1}[S,T];\\ &\, p_1\pow{0}[S,T], p_2\pow{1}[S,T]\in \Sigma_1[S,T], p_1\pow{1}[S,T]\in L^w (p_2\pow{1}[S,T])\} \\
& \ldots \dots
\end{align*}

Then, we have

$$ {\mc R}[S,T] \subset \Sigma_1[S,T] \subset \Sigma_2[S,T] \subset \cdots .$$
Now put,
\begin{align*}
& {\mc S}_1[S,T]= \Sigma_1[S,T]\setminus {\mc R}[S,T] \\
& {\mc S}_2[S,T]= \Sigma_2[S,T]\setminus \Sigma_1[S,T] \\
& \cdots \cdots \\
& {\mc S}_{n+1}[S,T]= \Sigma_{n+1}[S,T]\setminus \Sigma_n[S,T].
\end{align*}
Then, one has
$${\mc S}[S,T]= \bigcup_{k=1}^\infty {\mc S}_k[S,T].$$
Thus, the following statement holds.
\betheo Let $ S,T\in \LDH$ satisfy the commutation relation $[S,T]=\Id_{\D}$ in weak sense and $S^*=\ov{S\ad}$ $T^*=\ov{T\ad}$. Then, the partial O*-algebra ${\mc M}[S,T]$ generated by $S$ and $T$ decomposes into a regular part ${\mc R}[S,T]$ and a singular part ${\mc S}[S,T]$, as described above.
\entheo

\section{Uncertainty relations and other applicative aspects}\label{sec_5}
If $A \in \LDH$ and $z \in {\mb C}$, we define, for $\xi \in \D$ with $\|\xi\|=1$,
$$ (\Delta A)_\xi(z)= \|(A-z\Id_D)\xi\|.$$

We notice that for $z=\ip{A\xi}{\xi}$, we obtain $$(\Delta A)_\xi:= (\Delta A)_\xi(\ip{A\xi}{\xi})=\left(\ip{A\xi}{A\xi}-|\ip{A\xi}{\xi}|^2\right)^{1/2}$$
which reduces to the well known {\em uncertainty} of $A$ when $A$ is selfadjoint and $\xi \in D(A^2)$.

Let now $S,T\in \LDH$ satisfy (CR.2) Then, as it is easily seen, $S-z\Id$ and $T-w\Id$ also satisfy (CR.2) and a simple application of the Cauchy-Schwarz inequality gives, for every $\xi\in \D$, with $\|\xi\|=1$,
$$ \ip{\xi}{\xi} \leq 2 \max\{\|(S-z\Id_\D)\xi\|, \|(S\ad-\overline{z}\Id_\D)\xi\|\} \max\{\|(T-w\Id_\D)\xi\|, \|(T\ad-\overline{w}\Id_\D)\xi\|\};$$
or,
\begin{equation}\label{eq_uncertainty}\ip{\xi}{\xi}\leq 2\max\{(\Delta S)_\xi(z), (\Delta S\ad)_\xi(\overline{ z})\}\max\{(\Delta T)_\xi(w), (\Delta T\ad)_\xi(\overline{w}) \}.\end{equation}

The latter inequality reads as an {\em uncertainty principle} for non necessarily selfadjoint operators satisfying (CR.2).

In the case of symmetric operators, if one of them is bounded from below it is known that there exists no vector $\xi$ for which the previous inequality becomes an equality (Arai \cite{arai}).
We discuss here a similar result concerning non self-adjoint (or even symmetric) operators, whose proof will be based, rather than on a general theorem, on explicit counterexamples.

The first step consists in finding a reasonable counterpart of the Heisenberg uncertainty relation for non commuting operators which are not necessarily self-adjoint. One possibility is given, as we have seen, by the inequality \eqref{eq_uncertainty}.
But this is not the only possible choice of uncertainty relation for $S$ and $T$. As a matter of fact, it is natural to consider the following second possibility: let $\alpha_S$ and $\alpha_T$ be two arbitrary (but fixed) real quantities and let $D=\alpha_S(S+S^\dagger)$ and $E=i\alpha_T(T-T^\dagger)$. Hence $D=D^\dagger$ and $E=E^\dagger$. Let us now suppose that (in the usual weak sense)
$$
[S^\dagger,T]-[S,T^\dagger]=0.
$$
Hence we have $[D,E]=iC$, where $C=2\alpha_S\alpha_T\Id=C^\dagger$, so that the standard Heisenberg uncertainty relation allows us to conclude that $\Delta D_\varphi\Delta E_\varphi\geq\frac{|<C>|}{2}=|\alpha_S\alpha_T|$. But, using the estimates $(\Delta D)_\varphi\leq|\alpha_S|\left((\Delta S)_\varphi+(\Delta S)_\varphi^\dagger\right)$ and $(\Delta E)_\varphi\leq|\alpha_T|\left((\Delta T)_\varphi+(\Delta T)_\varphi^\dagger\right)$, we conclude that, for all $\varphi\in\D$,
\be
\left((\Delta S)_\varphi+(\Delta S)_\varphi^\dagger\right)\left((\Delta T)_\varphi+(\Delta T)_\varphi^\dagger\right)\geq 1.
\label{32}\en
>From now on we will call UR$_1$ condition \eqref{eq_uncertainty} and UR$_2$ condition (\ref{32}) ({\em first and second uncertainty relation}). We will now show that these are really unequivalent conditions.

\vspace{2mm}

\berem It can be useful to extend the above inequalities to the case in which $[S,T]=C$, with $C\neq\Id$, in general. These extensions look like
$$
\max\{(\Delta S)_\varphi,(\Delta S)_\varphi^\dagger\}   \max\{(\Delta T)_\varphi,(\Delta T)_\varphi^\dagger\}   \geq\frac{\left|\left<\varphi,C\varphi\right>\right|}{2},
$$ and
$$
\left((\Delta S)_\varphi+(\Delta S)_\varphi^\dagger\right)\left((\Delta T)_\varphi+(\Delta T)_\varphi^\dagger\right)\geq \left|\Re\{\left<\varphi,C\varphi\right>\}\right|.
$$
\enrem

Let us start with showing that a vector which saturates UR$_1$ does not necessarily saturate UR$_2$. Let $a$ and $a^\dagger$ be such that $[a,a^\dagger]=\Id$. We consider the operators $S=\frac{1}{\sqrt{2}}\left(a+ia^\dagger\right)$ and $T=\frac{1}{\sqrt{2}}\left(a^\dagger+ia\right)$. It is clear that $[S,T]=\Id$, $S\neq T^\dagger$, and that $[S^\dagger,T]-[S,T^\dagger]=0$. Let $\Phi(z)$ be the coherent state of $a$, \cite{gazbook}: $a\Phi(z)=z\Phi(z)$, for all $z\in\Bbb{C}$. An easy computation shows that $(\Delta S)_\varphi=(\Delta S)_\varphi^\dagger=(\Delta T)_\varphi=(\Delta T)_\varphi^\dagger=\frac{1}{\sqrt{2}}$. Hence UR$_1$ is saturated by $\Phi(z)$ while UR$_2$ is satisfied (this is obvious) but not saturated.

On the other hand, let us now show an example of a vector which saturates UR$_2$ but not UR$_1$. For that, it is enough to take $S=a$ and $T=a^\dagger$. Once again, we compute the uncertainty of these operators on $\Phi(z)$, getting
$(\Delta S)_\varphi=(\Delta T)_\varphi^\dagger=0$ and $(\Delta S)_\varphi^\dagger=(\Delta T)_\varphi=1$. Therefore UR$_2$ is saturated while UR$_1$ is not.

These examples show that UR$_1$ and UR$_2$ are really different conditions and that both could be considered as uncertainty relations. However we should also say that UR$_2$ is relevant only if $T\neq T^\dagger$ and $S\neq-S^\dagger$ because, if this is not the case, then $E=0$ or $D=0$ so that (\ref{32}) is trivialized, while UR$_1$ holds also for self-adjoint $S$ and $T$. We also notice that the two inequalities are different also for finite-dimensional Hilbert spaces. Let us consider, indeed, $\h=\Bbb{C}^2$. If we take
$$
S=\left(
    \begin{array}{cc}
      0 & s \\
      0 & 0 \\
    \end{array}
  \right),\quad\mbox{and}\quad T=\left(
    \begin{array}{cc}
      0 & 0 \\
      q & 0 \\
    \end{array}
  \right),
$$
with real $s$ and $q$, straightforward computations show that $(\Delta S)_\varphi=|s||\varphi_2|^2$, $(\Delta S)^\dagger_\varphi=|s||\varphi_1|^2$, $(\Delta T)_\varphi=|q||\varphi_1|^2$ and $(\Delta T)^\dagger_\varphi=|q||\varphi_2|^2$, where $\varphi=(\varphi_1,\varphi_2)$, $|\varphi_1|^2+|\varphi_2|^2=1$. We conclude that UR$_2$ is saturated only if $\max\left(|\varphi_1|^2,|\varphi_2|^2\right)=\sqrt{\frac{1}{2}\left||\varphi_1|^2-|\varphi_2|^2\right|}$ , which is always false. On the other hand, UR$_1$ is saturated if $1=\left||\varphi_1|^2-|\varphi_2|^2\right|$, which is always  true if $\varphi_1=1$ or if $\varphi_2=1$.

\subsection{The Swanson model}

We will now show, in a concrete model,  that UR$_1$ and UR$_2$ need not be satisfied at once by a certain vector. This is our version of the analogous result discussed in \cite{arai}.

Let $a$ and $a^\dagger$ be such that $[a,a^\dagger]=\Id$, and let us introduce $S=a\cos(\theta)+ia^\dagger\sin(\theta)$ and $T=a^\dagger\cos(\theta)+ia\sin(\theta)$. Here $\theta$ is an angle which, in the Swanson model, is assumed to belong to the open interval $\left]-\frac{\pi}{4},\frac{\pi}{4}\right[$, \cite{bagpb4}. However, this limitation is not useful here and will not be assumed. Let $\varphi$ be a generic vector in $\D$, and let us introduce the  quantities $C_\varphi=\left<\varphi,a^\dagger a\varphi\right>-\left|\left<\varphi,a\varphi\right>\right|^2$ and $E_\varphi=\Im\left\{\left<\varphi,{a^\dagger}^2 \varphi\right>-\left<\varphi,a^\dagger\varphi\right>^2\right\}$. Hence we find that
$$
\left\{
\begin{array}{ll}
((\Delta S)_\varphi)^2=C_\varphi+\sin^2(\theta)-\sin(2\theta)E_\varphi,\\ ((\Delta S)_\varphi^\dagger)^2=C_\varphi+\cos^2(\theta)-\sin(2\theta)E_\varphi,\\
((\Delta T)_\varphi)^2=C_\varphi+\cos^2(\theta)+\sin(2\theta)E_\varphi,\\ ((\Delta T)_\varphi^\dagger)^2=C_\varphi+\sin^2(\theta)+\sin(2\theta)E_\varphi.\\
\end{array}
\right.
$$
We want to show now that, for particular choices of $\theta$, there exists no state which saturates UR$_1$ or UR$_2$.

Let us first consider $\theta=0$. Hence UR$_2$ is saturated if and only if $C_\varphi=0$, which holds true if, for instance, $\varphi$ is the vacuum of $a$ ($a\varphi=0$) or if $\varphi$ is a coherent state for $a$ ($a\varphi=z\varphi$, for some $z\in\Bbb{C}$ related to $\varphi$). On the other hand, it turns out that UR$_1$ is saturated if and only if $C_\varphi=-\frac{1}{2}$, which is never satisfied since $C_\varphi$ is always (i.e. for all possible $\varphi\in\D$) greater or equal to zero.

Let us now fix $\theta=\frac{\pi}{4}$. In this case UR$_1$ is saturated if $\varphi$ is such that $\sqrt{(C_\varphi+\frac{1}{2})^-E_\varphi^2}=\frac{1}{2}$, while $\varphi$ saturates UR$_2$ if $\sqrt{(C_\varphi+\frac{1}{2})^-E_\varphi^2}=\frac{1}{4}$. Hence, suppose that a vector $\varphi_0$ exists in $\D$ saturating the UR$_2$. For such vector, then, we should have $\sqrt{(C_{\varphi_0}+\frac{1}{2})^-E_{\varphi_0}^2}=\frac{1}{4}$, which is less than $\frac{1}{2}$. Hence $\varphi_0$ does not satisfies UR$_1$, which is impossible. Hence such a vector cannot exist.


\section*{Acknowledgements}

This work was partially supported by the Japan Private School Promotion Foundation and partially by CORI, Universit\`a di Palermo. F.B. and C.T. acknowledge
the warm hospitality of the Department of Applied Mathematics of the Fukuoka University. A.I. acknowledges the hospitality of the Dipartimento di Matematica
e Informatica, Universit\`a di Palermo.


\begin{thebibliography}{99}
\bibitem{ait_book}  {   J-P. Antoine, A. Inoue, and C. Trapani},
{\it Partial *-Algebras and Their Operator Realizations}, Kluwer, Dordrecht, 2002.
\bibitem{arai} A. Arai, {\em Generalized Weak Weyl Relation and Decay of Quantum dymamics}, Rev. Math. Phys. {\bf 17} (2005) 1071-1109.
\bibitem{bagpb1} F. Bagarello, {\em Pseudo-bosons, Riesz bases and coherent states}, J. Math. Phys., {\bf 50}, DOI:10.1063/1.3300804, 023531 (2010) (10pg).
    \bibitem{bagpb4} F. Bagarello, {\em Examples of Pseudo-bosons in quantum mechanics},  Phys. Lett. A, {\bf 374}, 3823-3827 (2010)
\bibitem{bagrep} F. Bagarello, {\em Pseudo-bosons, so far}, Rep. Math. Phys., in press
\bibitem{barutracza} A. O. Barut and R. Racza, {\em Theory of Group representations and Applications}, PWN Warszawa, 1980.

    \bibitem{fuglede} B. Fuglede, {\em On the relation $PQ-PQ=-iI$}, Math. Scand. {\bf 20} (1967), 79-88.
    \bibitem{gazbook} J.-P. Gazeau, {\em Coherent states in quantum physics}, Wiley-VCH, Weinheim (2009)
\bibitem{yosida} K. Yosida, {\em Functional Analysis}, Springer Verlag, Berlin (1990)

    \bibitem{schmud_hei1} K. Schm\"udgen, {\em On the Heisenberg commutation relations I}, J. Funct. Anal. {\bf 50} (1983), 8-49.
        \bibitem{schmud_hei2} K. Schm\"udgen, {\em On the Heisenberg commutation relations II}, Publ. RIMS, Kyoto Univ. {19} (1983), 601-671.
\bibitem{schmudgen} K. Schm\"{u}dgen,
{\em Unbounded Operator Algebras and Representation Theory},
Birkh\"{a}user-Verlag, Basel, 1990.
\bibitem{reedsimon} M. Reed, B. Simon, {\em Methods of modern mathematical physics}, vol. I, {\em Functional Analysis}, Academic Press, San Diego (1980)






\end{thebibliography}
\end{document}